\pdfoutput=1
\documentclass{llncs}
\usepackage{makeidx}  
\usepackage{amssymb}

\usepackage{enumerate}
\usepackage{url}
\usepackage{color}
\usepackage{mathtools}
\usepackage[margin=1in]{geometry}

\begin{document}

\DeclarePairedDelimiter\ceil{\lceil}{\rceil}
\DeclarePairedDelimiter\floor{\lfloor}{\rfloor}

\frontmatter          
\pagestyle{headings}  


\mainmatter              

\title{Finite automata with advice tapes}


\author{ U\u{g}ur K\"{u}\c{c}\"{u}k\inst{1} \and A. C. Cem Say\inst{1} \and  Abuzer Yakary{\i}lmaz\inst{2,}\thanks{Abuzer Yakary{\i}lmaz was partially supported by ERC Advanced Grant MQC.}}

\authorrunning{U\u{g}ur K\"{u}\c{c}\"{u}k et al.} 

\tocauthor{ U\u{g}ur K\"{u}\c{c}\"{u}k,  A. C. Cem Say, and  Abuzer Yakary{\i}lmaz}

\institute{Bo\u{g}azi\c{c}i University, Istanbul, Turkey\\
\email{ugur.kucuk@boun.edu.tr, say@boun.edu.tr}
 \and University of Latvia, R\={\i}ga, Latvia\\
 \email{abuzer@lu.lv} }


\maketitle

\begin{abstract}
We define a model of advised computation by finite automata where the advice is provided on a separate tape. We consider several variants of the model where  the advice is deterministic or randomized, the input tape head is  allowed real-time, one-way, or two-way access,  and the automaton is classical or quantum. We prove several separation results among these variants, demonstrate an infinite hierarchy of language classes recognized by automata with increasing advice lengths, and establish the relationships between this and the previously studied ways of providing advice to finite automata. 
\end{abstract}

\section{Introduction}
Advised computation is based on the idea of providing external trusted assistance, depending only on the length of the input, to a computational device in order to extend its capability for solving certain problems \cite{KL82}. Work on advised finite automaton models started with \cite{DH95}, where the advice string is prefixed to the input tape, and continued with a sequence of papers starting with \cite{TYL10}, where the automaton reads the advice  in parallel with the input from a separate track.

In this paper, we propose a new architecture for advised finite-state computation which enables the automata to use the advice more flexibly than the setups mentioned above.  The idea is simply to let the machine use a separate one-way tape for the advice, thereby enabling it to pause on the input tape while processing the advice, or vice versa. (Examples of finite-state machines with such a separate tape for \textit{untrusted} advice can  be seen in \cite{DS92}.) Our model differs from an alternative proposal of Freivalds for advised finite-state automata \cite{Fr10} in the number of allowed advice tapes, and the way in which the advice can be accessed. We consider many variants of our machines, where the advised automaton is classical or quantum, the tapes can be accessed in various alternative modes, and the advice is deterministic or randomized. The power of these variants are compared among themselves, and also with the corresponding instances of the alternative models in the literature. 

\section{Previous work}
\label{sec:previouswork}

Finite automata that take advice were first examined by Damm and Holzer \cite{DH95}. In their model, the advice string, which depends only on the length of the input, is placed on the input tape so that it precedes the original input. 
We call such a machine a \textit{finite automaton with advice prefix}. The automaton simply reads the advice first, and then goes on to scan the input.  Damm and Holzer studied $\mathsf{REG/}const$, which is the class of languages that can be recognized by real-time deterministic finite automata that use constant-length advice, and showed that letting the advice string's length to be an increasing function of the input string's length, say, a polynomial, does not enlarge the class of  languages recognized by such automata within this setup. They also  used Kolmogorov complexity arguments to prove that every additional bit of advice extends the class of languages that can be recognized by finite automata in this model, that is, $\mathsf{REG/}(k-1)\subsetneq\mathsf{REG/}k$, for all $k \geq 1$.

Another model of advised finite automata was examined by Tadaki et al. in \cite{TYL10}, and later by T. Yamakami in \cite{Ya08,Ya10,Ya11,Ya12}. This setup enables the automata to process the advice in parallel with the input, by simply placing the advice in a separate track of the input tape.  
In this manner, an advice string of length $n$ can be provided, and meaningfully utilized, for inputs of length $n$. This enhances the language recognition power, as can be seen by considering the relative ease of designing such a \textit{finite automaton with  advice track} for the language $\{a^nb^n|\  n \in \mathbb{N}\}$, which 
can not be recognized by any finite automaton with  advice prefix. Yamakami studied variants of this model with probabilistic and quantum automata, and randomized advice \cite{Ya10,Ya12}, and provided characterizations of the related classes of languages. Note that the track structure in this model both limits the length of the advice by the length of the input, and forces the advice to be scanned synchronously with the input.

R. Freivalds formulates and studies yet another model of advised finite automata in \cite{Fr10,AF10}.  Freivalds' model incorporates one or more separate tapes for the advice to be read from. Both the input and  the advice tapes have two-way heads.  
Unlike the previously mentioned models, the advice string for inputs of length $n$ are supposed to be useful for all shorter inputs as well, and some negative results depend on this additional requirement.

\section{Our model}
\label{sec:ourmodel}

We model advice as a string provided on a separate read-only tape. As usual, the content of the advice depends only on the length of the input. Formally, the advice to the automaton is determined by an advice function $h$, which is a mapping from $\mathbb{N}$ to strings in $\Gamma^*$, where $\Gamma$ is the advice alphabet. This function may or may not be computable.

Our advised machine model is then simply a finite automaton with two tapes. 
The transition function of  a (two-way) \textit{deterministic finite automaton with  advice tape} (dfat) determines the next move of the machine based on the current internal state, and the symbols scanned by the input and advice tape heads. Each move specifies the next state, and a head movement direction (right,  left, or stay-put) for each tape. A tape head that is allowed to move in all these directions is called \textit{two-way}. A head that is not allowed to move left is called \textit{one-way}. We may also require a head to be \textit{real-time}, forcing it to move to the right at every step. As will be shown, playing with these settings changes the computational  power of the resulting model. We assume that both the input and the advice strings are delimited by special end-marker symbols, beyond which the automaton does not attempt to move its heads. The machine halts and announces the corresponding decision when it enters one of the two special states $q_{accept}$ and $q_{reject}$. 

Unlike Freivalds \cite{Fr10}, we do not allow two-way motion of the advice tape head, as permitting this head to make leftward moves would cause ``unfair'' accounting of the space complexity of the advised machine.\footnote{See Section 5.3.1 of \cite{Go08} for a discussion of this issue in the context of certificate tape heads.} 

A language $L$ is said to be recognized by such a dfat $M$ using $O(f(n))$-length advice if there exists an advice function $h$ with the following properties:
\begin{itemize}
\item $|h(n)| \in O(f(n))$ for all $n \in \mathbb{N}$, and,
\item $M$ eventually halts and accepts when started with the input tape containing a string $x$ of length $n$, and the advice tape containing $h(n)$, if and only if $x \in L$.
\end{itemize} 

We need a notation for talking about  language families corresponding to different settings of the tape access modes and advice lengths. We will use the template ``$\mathsf{CLASS}/f(n)(\mathtt{specification\ list})$" for this purpose. In that template, the name of the complexity class corresponding to the unadvised, two-way version of the automaton in question will appear as the $\mathsf{CLASS}$ item. The function description $f(n)$ will denote that the machine uses advice strings of length $O(f(n))$ for inputs of length $n$. (General descriptors like $poly$ and $exp$, for  polynomial and exponential bounds, respectively, will also be used.) Any further specifications about, for instance, additionally restricted head movements, will be given  in the list within the final parentheses. For example, the class of languages recognized by dfat's with real-time input and one-way advice tapes that use linear amounts of advice will be denoted $\mathsf{SPACE(1)}/n(\mathtt{rt\mbox{-}input})$.\footnote{Although $\mathsf{SPACE(1)}$ is well known to equal the regular languages, we avoid the shorter notation $\mathsf{REG}/n$, which was used for the advice track model, and which will turn out to represent a strictly smaller class.}

We will also be examining randomized advice, as defined by Yamakami \cite{Ya10}. In this case, the advice is randomly selected from a set of alternatives according to a pre-specified probability distribution. Deterministic finite automata which use randomized advice can perform tasks which are impossible with deterministic advice \cite{Ya10}. The use of randomized advice will be indicated by the letter $R$ appearing before the advice length in our class names. We will use an item in the parenthesized specification list to indicate whether bounded or unbounded error language recognition is intended, when this is not clear from the core class name. 

We define the probabilistic and quantum versions of our advised automata by generalizing the definition for deterministic automata in the standard way, see, for instance, \cite{YS11A}. The transition function of a \textit{probabilistic finite automaton with advice tape} (pfat) specifies not necessarily one, but possibly many choices, associated with selection probabilities, for the next move at every step, with the well-formedness condition that the probabilities of these choices always add up to 1. In the case of \textit{quantum finite automata with advice tapes} (qfat's), each such choice is associated not with a probability, but with an amplitude (a real number in the interval [-1,1]). The  presentation of our results on qfat's will not require knowledge of technical details of their definitions such as well-formedness conditions, and we therefore omit these for space constraints, referring the reader to \cite{YS11A}. We should stress that there are many mutually inequivalent quantum finite automaton definitions in the literature, and we use the most powerful one \cite{Hi10,YS11A}. The quantum machines with advice tracks defined in \cite{Ya12} are based on an older model \cite{KW97}, and this difference will be significant in our discussion in Section \ref{sec:quantum}.

The notational convention introduced above is flexible enough to represent the language classes corresponding to the probabilistic and quantum advised machines as well. $\mathsf{BQSPACE(1)}/n(\mathtt{rt\mbox{-}input,  rt\mbox{-}advice })$, for instance,  denotes the class of languages recognized with bounded error by a qfat using linear-length advice, and real-time input and advice heads.

The model of real-time finite automata with advice tracks \cite{TYL10} is equivalent to our  model with a separate advice tape when we set both the input and advice tape heads to be real-time. Therefore, all the results shown for the advice track model are inherited for this setting of our machines. 
For instance, $\mathsf{SPACE(1)}/n(\mathtt{rt\mbox{-}input,rt\mbox{-}advice})=\mathsf{REG}/n$, 
where $\mathsf{REG}/n$ is defined in \cite{TYL10}.
On the other hand, the quantum class  $\mathsf{1QFA}/n$ of \cite{Ya12}  does \textit{not} equal $\mathsf{BQSPACE(1)}/n(\mathtt{rt\mbox{-}input,rt\mbox{-}advice})$, as we will show in Section \ref{sec:quantum}.

Note that we allow only one advice tape in our model. This is justified by the following observation about the great power of one-way finite automata with multiple advice tapes.

\begin{theorem}
 Every language can be recognized by a finite automaton with a one-way input tape and two one-way advice tapes.
\end{theorem}
\begin{proof}
Let $L$ be any language on the alphabet $\Sigma$. We construct a finite automaton $M$ that recognizes $L$ using a one-way input tape and two one-way advice tapes as follows.

Let $\Gamma = \Sigma \cup \{c_a,c_r\}$ be the advice alphabet, where $\Sigma \cap \{c_a,c_r\} = \emptyset$. For an input of length $n$, the advice on the first advice tape lists every string in $\Sigma^n$ in  alphabetical order, where every member of $L$ is followed by a  $c_a$, and every nonmember is followed by a $c_r$.  So the content of the first advice tape looks like $w_1c_1w_2c_2 \cdots w_{|\Sigma|^n}c_{|\Sigma|^n}$, where $w_i \in \Sigma^n$, and $c_i \in \{c_a,c_r\}$ for $i \in \{1,\ldots,|\Sigma|^n \}$.

The second advice tape  content looks like ``$c_ac_r^nc_ac_r^n \cdots c_ac_r^nc_a$'', with $|\Sigma|^n$ repetitions, and will be used by the machine for counting up to $n+1$ by moving between two consecutive $c_a$ symbols on this tape.

$M$ starts its computation while scanning the first symbols of the input string and $w_1$ on the first advice tape. It attempts to match the symbols it reads from the input tape and the first advice tape, moving synchronously on both tapes.  If the $i$th input symbol does not match the $i$th symbol of $w_j$, $M$ pauses on the input tape, while moving the two advice heads simultaneously until the second advice head reaches the next  $c_a$, thereby placing the first advice tape head on the $i$th position of $w_{j+1}$, where $1\leq i \leq n$, and $1 \leq j < |\Sigma|^n$. As the words on the first advice tape are ordered lexicographically, it is  guaranteed that  $M$ will eventually locate the word on the first advice tape that matches the input in this manner. $M$ halts when it sees the endmarker on the input tape, accepting if the symbol read at that point from the first advice tape is $c_a$, and rejecting otherwise. 
\end{proof}

\section{Deterministic finite automata with advice tapes}

It is clear that a machine with  advice tape is at least as powerful as a machine of the same type with  advice track, which in turn is superior to a corresponding machine with  advice prefix, as mentioned in Section \ref{sec:previouswork}. We will now show that allowing either one of the input and advice head to pause on their tapes does enlarge the class of recognized languages.

\begin{theorem}\label{theorem:equal}
$\mathsf{REG}/n \subsetneq \mathsf{SPACE(1)}/n(\mathtt{rt\mbox{-}input})$.  
\end{theorem}
\begin{proof}
It follows trivially from the definitions of the classes that
\begin{center}
$\mathsf{REG}/n = \mathsf{SPACE(1)}/n(\mathtt{rt\mbox{-}input, rt\mbox{-}advice}) \subseteq \mathsf{SPACE(1)}/n(\mathtt{rt\mbox{-}input}).$ 
\end{center}
Let $|w|_{\sigma}$ denote the number of occurrences of symbol $\sigma$ in string $w$.
To show that the above subset relation is proper, we will consider the language 
$\mathtt{EQUAL_2} = \{w|\ w\in \{a,b\}^*\ \mbox{and}\ |w|_a = |w|_b\}$, which is known \cite{TYL10} to lie outside $ \mathsf{REG}/n $.

One can construct a finite automaton that recognizes $\mathtt{EQUAL_2}$ with real-time input and one-way access to linear advice as follows. For inputs of odd length, the automaton rejects the input. For inputs of even length, $n$, the advice function is $h(n) = a^{n/2}$. The automaton moves its advice head one position to the right for each $a$ that it reads on the input. The input is accepted if the number of $a$'s on the two tapes match, and rejected otherwise. 
\end{proof}

\begin{theorem} 
$\mathsf{REG}/n \subsetneq \mathsf{SPACE(1)}/n(\mathtt{1w\mbox{-}input,rt\mbox{-}advice})$.
\label{THM_EQUAL_IN_1W_RT_LINEAR}
\end{theorem}

\begin{proof}
Consider the language $\mathtt{EQUAL} = \{w| w\in \{a,b,c\}^*$ where $|w|_a=|w|_b \}$, which is similar to $\mathtt{EQUAL_2}$, but with a bigger alphabet. $\mathtt{EQUAL} \notin \mathsf{REG}/n$, as can be shown easily by Yamakami's characterization theorem (Th. 2 of \cite{Ya10}) for this class.
We will describe a dfat $M$ with one-way input, and real-time access to an advice string that is just $a^{2n}$, where $n$ is the input length.

$M$ moves the advice head one step to the right for each $a$ that it scans in the input. When it scans a $b$, it advances the advice head by three steps and for each $c$, scanned on the input tape, the advice head is moved two steps. If the advice head attempts to move beyond the advice string, $M$ rejects. When the input tape head reaches the end of the tape, $M$ waits to see if the advice tape head will also have arrived at the end of the advice string after completing the moves indicated by the last input symbol. If this occurs, $M$ accepts, otherwise, it rejects.

Note that the advice head is required to move exactly $|w|_a+3|w|_b+2(n-|w|_a-|w|_b)$ steps, which equals $2n$ if and only if the input is a member of $\mathtt{EQUAL}$.
\end{proof}

Tadaki et al. \cite{TYL10} studied one-tape linear-time Turing machines with an advice track, and showed  that the class of languages that they can recognize coincides with $\mathsf{REG}/n$. Theorem \ref{theorem:equal} above lets us conclude that simply having a separate head for advice increases the computational power of a real-time dfa, whereas the incorporation of a single two-way head for accessing both advice and a linear amount of read/write memory simultaneously does not.

As noted earlier, advice lengths that are increasing functions of the input length are not useful in the advice prefix model. Only linear-sized advice has been studied in the context of the advice track model \cite{TYL10,Ya10}. Theorem \ref{theorem:akbmck} demonstrates a family of languages for which very small increasing advice length functions are useful in the advice tape model, but not in the advice track model.

\begin{theorem}\label{theorem:akbmck}
For every function $f:\mathbb{N}\rightarrow\mathbb{N}$ such that  $f(n) \in \omega(1) \cap O(\sqrt{n})$, $\mathsf{SPACE(1)}/f^2(n)(\mathtt{1w\mbox{-}input}) \nsubseteq \mathsf{REG}/n$.


\end{theorem}
\begin{proof}
Consider the language $\mathtt{L_f}=\{a^kb^mc^k|  k \leq f(n),   n = k + m + k\}$, for any function $f$ satisfying the properties in the theorem statement.

Theorem 2 of \cite{Ya10} can be used to show  $\mathtt{L_f} \notin \mathsf{REG}/n $.

One can construct a dfat with one way access to input and advice that recognizes $\mathtt{L_f}$ as follows. For inputs of length $n$, the advice string is of the form $\#\#a\#aa\#aaa\# \cdots \#a^{f(n)} \#$, with length $O(f^2(n))$. During any step, if the automaton detects that the input is not of the form $a^*b^*c^*$, it rejects the input. For each $a$ that it reads from the input tape, the automaton moves the advice tape head to the next $\#$ on the advice tape. (If the advice ends when looking for a $\#$, the input is rejected.) When the input tape head scans the $b$'s, the advice tape head remains idle. Finally, when the input head starts to scan the $c$'s, the automaton  compares the number of $c$'s on the input tape with the number of $a$'s that it can scan until the next $\#$ on the advice tape. If these  match, the input is accepted; otherwise it is rejected.
\end{proof}

When restricted to constant size advice, the parallelism and the two-way input access inherent in our model does not make it superior to the advice prefix model. As we show now, one can always read the entire advice before starting to read the input tape without loss of computational power in the constant-length advice case:

\begin{theorem}
For every  $k \in \mathbb{N}$, $\mathsf{SPACE(1)}/k=\mathsf{REG}/k$.
\end{theorem}
\begin{proof}
The relation $\mathsf{REG}/k\subseteq \mathsf{SPACE(1)}/k$ is trivial, since an automaton taking constant-length advice in the prefix or track formats can  be converted easily to one that reads it from a separate tape. For the other direction, note that a dfat $M$ with two-way input that uses $k$ bits of advice corresponds to a set S of $2^k$ real-time dfa's \textit{without} advice, each of which can be obtained by hard-wiring a different advice string to the program of $M$, and converting the resulting two-way dfa to the equivalent real-time machine, which exists by \cite{Sh59}. The advice string's job is just to specify which of these machines will run on the input string. It is  then easy to build a dfa with advice prefix which uses the advice to select the appropriate program to run on the input. 
\end{proof}

Since our model is equivalent to the advice prefix model for constant-length advice, we inherit the results like Theorem 5 of \cite{DH95}, which states that the longer advice strings one allows, the larger the class of languages we can recognize will be, as long as one makes sure that the advice and input alphabets are identical. 


For any language $L$ on an alphabet $\Sigma$, and for natural numbers $n$ and $k$ such that $k\le n$, we define the relation $\equiv_{L,n,k}$ on the set $\Sigma^k$ as follows:
$x \equiv_{L,n,k} y   \iff   \mbox{for all strings } z \mbox{ of length }n-k, xz
\in L\mbox{ if and only if }yz \in L.$
In the remainder of the paper, we will make frequent use the following lemma, which is reminiscent of Yamakami's characterization theorem for $\mathsf{REG}/n$ \cite{Ya10}, to demonstrate languages which are unrecognizable with certain amounts of advice by automata with one-way input.

\begin{lemma}\label{lemma:1waychar}
For any advice length function $f$, if $L \in \mathsf{SPACE(1)}/f(n)(\mathtt{1w\mbox{-}input})$, then  for all $n$ and all $k\le n$,  $\equiv_{L,n,k}$ has  $O(f(n))$ equivalence classes.
\end{lemma}
\begin{proof}
Let $M$ be the dfat which is supposed to recognize $L$ with an advice string of length $O(f(n))$. If we fix the position of the input head, there are just  $O(f(n))$ combinations of internal state and advice head position pairs that are potentially reachable for $M$. Assume that the number of equivalence classes of $\equiv_{L,n,k}$ is not $O(f(n))$. Then for some sufficiently large $n$, there exists two strings $x$ and $y$ of length $k$ in two different equivalence classes of $\equiv_{L,n,k}$ which cause $M$ to reach precisely the same head positions and internal state after being processed if they are presented as the prefixes of two $n$-symbol input strings in two separate executions of $M$. But $M$ will then have to give the same response to the two input strings $xz$ and $yz$ for any $z \in \Sigma^{n-k}$, meaning that $x \equiv_{L,n,k} y$. 
\end{proof}


We can now establish the existence of an infinite hierarchy of language classes that can be recognized by dfat's with increasing amounts of advice.

\begin{theorem}
For $  k \in \mathbb{Z^+} $, 
$  \mathsf{SPACE(1)}/n^{k}(\mathtt{1w\mbox{-}input})  \subsetneq \mathsf{SPACE(1)}/n^{k+1}(\mathtt{1w\mbox{-}input} ) $.
\end{theorem}

\begin{proof}
To prove the theorem statement, we will first define a family $\mathtt{L_k}$ of languages for $k \in \mathbb{Z^+}$, and then show  that advice strings of length $\Theta(n^i)$ are necessary (Lemma \ref{lemma:Li_Not_in_Ni-1}) and sufficient  (Lemma \ref{lemma:Li_in_Ni}) to recognize any particular member $\mathtt{L_i}$ of this family. 

\end{proof}

\begin{definition}
$\mbox{For}\  k \in \mathbb{Z^+}, \ \mathtt{L_k}= \{ c_k^{n_k} c_{k-1}^{n_{k-1}}\cdots c_1^{n_1} c_0^{n_0} c_1^{n_1}\cdots c_{k-1}^{n_{k-1}} c_k^{n_k} |  n_0 > 0 \ \mbox{and} \  n_j \ge 0 \ \mbox{for} \ j \in \{ 1,\ldots,k\}\}$ on the $k+1$-symbol alphabet $\{c_0, c_1,\ldots, c_k\}$. 
\end{definition}

\begin{lemma} \label{lemma:Li_Not_in_Ni-1}
$\mbox{For}\  i  \in \mathbb{Z^+},  \ \mathtt{L_i}\notin \mathsf{SPACE(1)}/n^{i-1}(\mathtt{1w\mbox{-}input})$
\end {lemma}

\begin{proof}
For a positive integer $n$, consider the set $S$ of strings of length $k =\floor{n/2}+1$ and of the form  $c_i^{*} c_{i-1}^{*} \cdots c_1^{*} c_0^{+}$. Note that each member of $S$ is the first half of a different member of $\mathtt{L_i}$, no two distinct members $x$ and $y$ of $S$ satisfy $x \equiv_{\mathtt{L_i},n,k} y$, and that there are $\Theta(n^i)$ members of $S$. We conclude using Lemma \ref{lemma:1waychar} that $\mathtt{L_i}\notin \mathsf{SPACE(1)}/n^{i-1}(\mathtt{1w\mbox{-}input})$.

\end{proof}

\begin{lemma} \label{lemma:Li_in_Ni}
	For $ i \in \mathbb{Z^+} $, $ \mathtt{L_i} \in \mathsf{SPACE(1)}/n^{i}(\mathtt{1w\mbox{-}input})$.
\end {lemma}

\begin{proof}
An inductive argument will be employed to show the truth of the statement, so let us first consider the language $\mathtt{L_1}$. To see that $\mathtt{L_1} \in \mathsf{SPACE(1)}/n^{1}(\mathtt{1w\mbox{-}input})$, we construct an advice function $h_1(n)$ and an automaton $M_1$ as follows. For inputs of length $n$, let $h_1(n) = 1^{n}$ be given as advice. The automaton $M_1$ checks if the input is of the form $ c_1^{i} c_0^{j} c_1^{k}$ for $i,k\ge 0$ and $j>0$. If not, it rejects. In parallel, $M_1$ moves the advice tape head while scanning the input as follows: For each  $c_1$ that comes before the first $c_0$ in the input, the advice tape head stays put. For each  $c_0$ in the input, the advice tape head moves one step to the right. Finally, for each  $c_1$ that comes after the last $c_0$ in the input, the advice tape head moves two steps to the right. The input is accepted if the endmarkers are scanned simultaneously on both tapes. Since the advice head moves exactly $j+2k$ steps, which equals $n=i+j+k$ if and only if $i=j$, we conclude that $M_1$ recognizes  $\mathtt{L_1}$ when provided with $h_1(n)$, a linear-length advice function.

Now let us prove that $\mathtt{L_i} \in \mathsf{SPACE(1)}/n^{i}(\mathtt{1w\mbox{-}input}) \implies \mathtt{L_{i+1}} \in \mathsf{SPACE(1)}/n^{i+1}(\mathtt{1w\mbox{-}input})$.

Assume $\mathtt{L_i} \in \mathsf{SPACE(1)}/n^{i}(\mathtt{1w\mbox{-}input})$. Then there should be a dfat $M_i$ which recognizes $\mathtt{L_{i}}$ when it has access to the advice function $h_i(n)$ of length $O(n^{i})$. Below, we construct a dfat $M_{i+1}$ and an advice function $h_{i+1}(n)$ of length $O(n^{i+1})$ such that  $M_{i+1}$ recognizes $\mathtt{L_{i+1}}$ when it has access to advice determined by function $h_{i+1}(n)$.

Note that the members of $\mathtt{L_{i+1}}$ are members of $\mathtt{L_{i}}$ sandwiched between equal numbers of $c_{i+1}$'s on each end. Therefore, the method for checking membership in $\mathtt{L_{i}}$  can be used in the test for   membership in $\mathtt{L_{i+1}}$ if one can  check whether the $c_{i+1}$ sequences at each end are of the same length separately. Based on this idea, we define the advice function $h_{i+1}(n)$ for  $\mathtt{L_{i+1}}$  in terms of the advice function $h_{i}(n)$ for  $\mathtt{L_{i}}$ as 
\[h_{i+1}(n)=  h_i(n)  \#_{i+1} h_i(n-2) c_{i+1} \#_{i+1} \cdots  \#_{i+1} h_i(n - 2\floor{\frac{n}{2}} )c_{i+1}^{\floor{\frac{n}{2}}}\#_{i+1},\]
that is, one concatenates all the strings $h_i(n-2j)c_{i+1}^j\#_{i+1}$ for $j \in \{0,\ldots,\floor{\frac{n}{2}}\}$ in increasing order, where $\#_{i+1}$ is a new symbol in $M_{i+1}$'s advice alphabet.
As $h_i(n)$ is of length $O(n^{i})$, the length of $h_{i+1}(n)$ can be verified to be  $O(n^{i+1})$.

When provided access to the advice function $h_{i+1}(n)$, the automaton $M_{i+1}$ performs the tasks below in parallel in order to recognize the language $\mathtt{L_{i+1}}$.

\begin{itemize}

\item The input is checked to be of the form $ c_{i+1}^{*} c_i^{*} \cdots c_1^{*} \ c_0^{+}\  c_1^{*} \cdots  c_i^{*} c_{i+1}^{*}$. If not, it is rejected.

\item For each $c_{i+1}$ on the input tape, that comes before any other symbol, the advice head is moved to the next $\#_{i+1}$ on the advice tape. If the endmarker is scanned on the advice tape at this step, the input is rejected. When the first non-$c_{i+1}$ symbol is scanned on the input, the control passes to the automaton $M_i$ for language $L_i$, which runs on the input tape content until the first $c_{i+1}$ or the endmarker, and uses as advice the content until the first $c_{i+1}$ or $\#_{i+1}$ on the advice tape. If  $M_i$ rejects its input, so does $M_{i+1}$. If $M_i$ accepts its input, $M_{i+1}$ accepts its input only if the number of $c_{i+1}$'s on the remainder of the input tape matches the number of  $c_{i+1}$'s on the advice tape until the first $\#_{i+1}$.

\end{itemize}

\end{proof}

We now show that  $\mathtt{PAL}$, the language of even-length palindromes on the alphabet $\{a,b\}$, is unrecognizable by dfat's with one-way input and polynomial-length advice:

\begin{theorem} \label{theorem:palnotin1w1w}
$\mathtt{PAL} \notin \mathsf{SPACE(1)}/poly(\mathtt{1w\mbox{-}input})$.
\end{theorem}

\begin{proof}
Similarly to the proof of Lemma \ref{lemma:Li_Not_in_Ni-1}, we  consider the set $S$ of all strings on $\{a,b\}$ of length $k=n/2$ for an even positive number $n$. No two distinct members $x$ and $y$ of $S$ satisfy $x \equiv_{\mathtt{PAL},n,k} y$, and there are $2^{\Theta(n)}$ members of $S$. We conclude using Lemma \ref{lemma:1waychar} that $\mathtt{PAL} \notin \mathsf{SPACE(1)}/poly(\mathtt{1w\mbox{-}input})$.

\end{proof}

Since a machine with real-time input does not have time to consume more than a linear amount of advice, we easily have
\begin{corollary}
For every function $f:\mathbb{N}\rightarrow\mathbb{N}$,
$\mathtt{PAL} \notin \mathsf{SPACE(1)}/f(n)(\mathtt{rt\mbox{-}input})$.
\label{CRL_PAL_NOT_IN_REST}
\end{corollary}

A natural question that arises during the study of advised computation is whether the model under consideration is strong enough to recognize $every$ desired language. The combination of two-way input tape head and exponentially long advice can be shown to give this power to finite automata. Let $\mathsf{ALL}$ denote the class of all languages on the input alphabet $\Sigma$.  

\begin{theorem} \label{theorem:all}
$\mathsf{SPACE(1)}/exp(\mathtt{rt\mbox{-}advice})=\mathsf{ALL}$.
\end{theorem}

\begin{proof}
The advice string for input length $n$ contains all members of the considered language of length $n$, separated by substrings consisting of $n+2$ blank symbols. The automaton compares the input with each of the strings listed on the advice tape. If it is able to match the input to a word on the advice tape, it accepts. If the advice ends without such a match, it rejects. Otherwise, the machine rewinds to the start of the input while consuming blanks from the next string on the advice tape. The advice length is $2^{O(n)}$.
\end{proof}

Whether $\mathsf{SPACE(1)}/exp(\mathtt{1w\mbox{-}input})=\mathsf{ALL}$ is an open question. If $\mathtt{PAL}\in \mathsf{SPACE(1)}/exp(\mathtt{1w\mbox{-}input})$ also remains unknown to us. But we are able to prove a separation between classes corresponding to machines with one-way versus two-way input that are confined to polynomial-length advice, as the following theorem shows.

\begin{theorem} 
$
	\mathsf{SPACE(1)}/poly(\mathtt{1w\mbox{-}input}) \subsetneq  \mathsf{SPACE(1)}/poly(\mathtt{2w\mbox{-}input})
$.
\end{theorem}


\begin{proof}
We already showed in Theorem \ref{theorem:palnotin1w1w} that polynomial-length advice is no help for dfat's with one-way input for recognizing $\mathtt{PAL}$. To prove the present theorem, we shall describe how a two-way dfa with real-time access to a quadratic-length advice string can recognize $\mathtt{PAL}$. On an input of length $n$, the advice tells the automaton to reject if $n$ is odd. For even $n$, the advice assists the automaton by simply marking the $n/2$ pairs $(i, n-i+1)$ of positions that should be holding matching symbols on the input string. Consider, for example $h(8) = \#10000001\#01000010\#00100100\#00011000\#$. The automaton should just traverse the input from the first symbol to the last while also traversing the part of the advice that lies between two separator symbols $(\#)$, and then do the same while going from the last symbol to the first, and so on. At each pass, the automaton should check whether the input symbols whose positions match those of the two $1$'s on the advice are identical. If this check fails at any pass, the automaton rejects the input, otherwise, it accepts.

The method described above requires a two way automaton with real-time access to an advice of length $n^2 / 2$. (The separator symbols are for ease of presentation, and are not actually needed for the construction.) 
\end{proof}

\section{Randomized advice for deterministic machines and vice versa}

We now turn to randomly selected probabilistic advice given to deterministic machines. Yamakami \cite{Ya10} proved that this setup yields an improvement in language recognition power over $\mathsf{REG}/n$, by demonstrating a deterministic automaton with advice track recognizing the center-marked palindrome language with randomized advice. Considering the amount of randomness involved in the selection of the advice string as a resource, Yamakami's example requires $O(n)$ random bits, since it requires picking a string from a set with $2^{O(n)}$ elements with uniform probability. Furthermore, reducing the error bound of Yamakami's automaton to smaller and smaller values requires extending the advice alphabet to bigger and bigger sizes. In the construction we will present in Theorem \ref{theorem:randadvice}, the number of random bits  does not depend on the input length, and any desired error bound can be achieved without modifying the advice alphabet.

\begin{theorem} \label{theorem:randadvice}
$ 	\mathsf{SPACE(1)}/n(\mathtt{1w\mbox{-}input})  \subsetneq \mathsf{SPACE(1)}/Rn(\mathtt{1w\mbox{-}input,bounded\mbox{-}error})
$.
\end{theorem}

\begin{proof}
We will use the language 
$\mathtt{EQUAL_3} = \{w|\ w\in \{a,b,c\}^*,  |w|_a = |w|_b = |w|_c \}$ to separate the language classes in the theorem statement. 

Let $k$ be any positive integer, $n=3k$, and consider the set $S$ of all strings of length $k$ and of the form  $a^{*}b^{*}c^{*}$. Note that $S$ has ${{k+2}\choose{2}}=\omega(n)$ members, and that  no two distinct members $x$ and $y$ of $S$ satisfy $x \equiv_{\mathtt{EQUAL_3},n,k} y$. We conclude using Lemma \ref{lemma:1waychar} that $\mathtt{EQUAL_3} \notin \mathsf{SPACE(1)}/n(\mathtt{1w\mbox{-}input,1w\mbox{-}advice})$.

To show that  $\mathtt{EQUAL_3} \in \mathsf{SPACE(1)}/Rn(\mathtt{1w\mbox{-}input,1w\mbox{-}advice, bounded \mbox{-}error})$, we will describe a set of advice strings, and show how a randomly selected member of this set can assist a one-way dfat $N$ to recognize $\mathtt{EQUAL_3}$ with overall bounded error. We shall be adapting a technique used by Freivalds in \cite{Fr79}.

 If the input length $n$ is not divisible by $3$, $N$ rejects.  If $n =3k$ for some integer $k$,  the advice is selected with equal probability from a collection of linear-size advice strings $A_i = 1^i\#1^{ki^2 + ki + k}$ for $i \in \{1,\ldots,s\}$, where $s$ is a constant. 

$N$ starts by reading the $1$'s in the advice string that precede the separator character $\#$, thereby learning the number $i$. $N$ then starts to scan the input symbols, and  moves the advice head $1$, $i$, or $i^2$ steps to the right for each $a$, $b$ or $c$ that it reads on the input tape, respectively. The input is accepted if the automaton reaches the ends of the input and advice strings simultaneously, as in the proof of Theorem \ref{THM_EQUAL_IN_1W_RT_LINEAR}. Otherwise, the input is rejected.

Note that the automaton accepts the input string $w$ if and only if the number of symbols in the advice string that comes after the separator symbol is equal to the total number of moves  made by the advice tape head while the input head scans $w$. $N$ accepts $w$ if and only if $|w|_a + |w|_bi + |w|_ci^2 = k + ki + ki^2$, which trivially holds for $w \in \mathtt{EQUAL_3}$ no matter which advice string is selected, since  $|w|_a = |w|_b = |w|_c = k$ in that case. 

If $w \notin \mathtt{EQUAL_3}$, the probability of acceptance is equal to the probability of selecting one of the roots of the quadratic equation  $(|w|_c -k) i^2 + (|w|_b -k)i +(|w|_a-k) = 0$ as the value of $i$. This probability is bounded by $\frac{2}{s}$, and can be pulled down to any desired level by picking a bigger value for $s$, and reorganizing the automaton accordingly. 

\end{proof}


Another way of integrating randomness to the original model is to employ probabilistic computation with access to deterministic advice. We show below that probabilistic automata with advice can recognize more languages with bounded error than their deterministic counterparts.

\begin{theorem} 
\label{theorem:prob_auto_det_advice}
$\mathsf{SPACE(1)}/n(\mathtt{1w\mbox{-}input}) \subsetneq\mathsf{BPSPACE(1)}/n(\mathtt{1w\mbox{-}input,bounded\mbox{-}error})$.
\end{theorem}

\begin{proof}
$\mathsf{SPACE(1)}/n(\mathtt{1w\mbox{-}input,1w\mbox{-}advice}) \subseteq \mathsf{BPSPACE(1)}/n(\mathtt{1w\mbox{-}input, 1w\mbox{-}advice,bounded\mbox{-}error})$ is by definition. So it remains to show that there is a language which can not be recognized by a one way dfat with one way access to linear-size advice but can be recognized by bounded error by a pfat with the same amount of advice. We claim that  $\mathtt{EQUAL_3}$, which was introduced and  shown to lie outside  $\mathsf{SPACE(1)}/n(\mathtt{1w\mbox{-}input,1w\mbox{-}advice})$ in the proof of Theorem \ref{theorem:randadvice}, is one such language. We now describe how to construct a one-way pfat $P$ and an associated linear-length advice function to recognize $\mathtt{EQUAL_3}$ for any specified nonzero error bound $\varepsilon < \frac{1}{2}$. The idea is reminiscent of that used for the proof of Theorem \ref{theorem:randadvice}. However we now specify a deterministic advice function which contains all the alternatives and let the probabilistic automaton randomly pick and use one.

Let $n$ denote the length of the input, and let $s=\lceil\frac{2}{\varepsilon}\rceil$. If $n$ is not divisible by $3$, the automaton rejects with probability $1$.  If $n$ is divisible by $3$, the advice is the string  $\#1^{n}\#1^{\frac{7n}{3}}\#\dots\#1^{\frac{n}{3}s^2 + \frac{n}{3}s + \frac{n}{3}}$, obtained by concatenating all the strings $\#1^{\frac{n}{3}i^2 + \frac{n}{3}i + \frac{n}{3}}$ for $i \in \{1,\ldots,s\}$ in increasing order.

$P$ starts by randomly picking an integer $i$ between 1 and $s$, and moving its advice head to the $i$'th $\#$. It then starts scanning the input, moving the advice head by $1$, $i$, or $i^2$ steps  for each $a$, $b$ or $c$, just as we had in the proof of Theorem \ref{theorem:randadvice}. It accepts if and only if the advice head reaches the next $\#$ (or the end of the advice string) simultaneously with the arrival at the end of the input. The correctness of the algorithm follows from the argument in the proof of Theorem \ref{theorem:randadvice}.

\end{proof}


\section{Quantum finite automata with advice tapes}\label{sec:quantum}
Yamakami \cite{Ya12} defined the class $\mathsf{1QFA}/n$ as the collection of languages which can be recognized by real-time Kondacs-Watrous quantum finite automata (KWqfa's) with advice tracks. The KWqfa is one of many inequivalent models of quantum finite-state computation that were proposed in the 1990's, and is known to be strictly weaker than classical finite automata in the context of bounded-error language recognition \cite{KW97}. This weakness carries over to the advised model of \cite{Ya12}, with the result that there exist some regular languages that are not members of $\mathsf{1QFA}/n$. We use a state-of-the-art model of quantum automaton that can simulate its classical counterparts trivially, \cite{Hi10,YS11A} so we have:
\begin{theorem}
$\mathsf{1QFA}/n\subsetneq\mathsf{BQSPACE(1)}/n(\mathtt{rt\mbox{-}input,rt\mbox{-}advice})$.
\end{theorem}

Whether this properly strong version of qfa can outperform its classical counterparts with advice tapes is an open question. We are able to show a superiority of quantum over classical in the following restricted setup, which may seem silly at first sight:
Call an advice tape \textit{empty} if it contains the standard blank tape symbol in all its squares. We say that a machine $M$ receives \textit{empty advice} of length $f(n)$, if it is just allowed to move its advice head on the first $f(n)$ squares of an empty advice tape, where $n$ is the input length. This restriction will be represented by the presence of the designation \texttt{empty} in the specification lists of the relevant complexity classes.
\begin{theorem}
$\mathsf{BPSPACE(1)}/n(\mathtt{rt\mbox{-}input,1w\mbox{-}empty\mbox{-}advice}) \\ 
\indent \indent \indent \indent \indent  \subsetneq\mathsf{BQSPACE(1)}/n(\mathtt{rt\mbox{-}input,1w\mbox{-}empty\mbox{-}advice})$.
\end{theorem}
\begin{proof}
	An empty advice tape can be seen as an increment-only counter, where each move of the advice tape head corresponds to an incrementation on the counter, with no mechanism for decrementation or zero-testing provided in the programming language. In \cite{YFSA11A}, Yakary{\i}lmaz et al.  studied precisely this model. It is obvious that classical automata augmented with such a counter do not gain any additional computational power, so  $\mathsf{BPSPACE(1)}/n(\mathtt{rt\mbox{-}input,1w\mbox{-}empty\mbox{-}advice})$ equals the class of regular languages, just like the corresponding class without advice. On the other hand, real-time qfa's augmented with such an increment-only counter were shown to recognize some nonregular languages like $ \mathtt{EQUAL_2} $ with bounded error in \cite{YFSA11A}.

\end{proof}

Since increment-only counters are known to increase the computational power of real-time qfa's in the unbounded-error setting as well, \cite{YFSA11A}, we can also state

\begin{theorem}
$\mathsf{PrSPACE(1)}/n(\mathtt{rt\mbox{-}input,1w\mbox{-}empty\mbox{-}advice}) \\ 
\indent \indent \indent \indent \indent  \subsetneq\mathsf{PrQSPACE(1)}/n(\mathtt{rt\mbox{-}input,1w\mbox{-} empty\mbox{-}advice})$.
\end{theorem}

\section{Open questions}
\begin{itemize}
\item Real-time probabilistic automata can be simulated by deterministic automata which receive coin tosses within a randomly selected advice string. 
It would be interesting to explore the relationship between  deterministic automata working with randomized advice, and probabilistic automata working with deterministic advice further.
\item Are there languages which cannot be recognized with any amount of advice by a dfat with one-way input? Does the answer change for pfat's or qfat's?
\item Can qfat's recognize any language which is impossible for pfat's with non-empty advice?
\end{itemize}

\section*{Acknowledgments}
We thank G\"{o}kalp Demirci, \"{O}zlem Salehi, and an anonymous reviewer for an earlier version of this manuscript for their helpful comments.

\bibliographystyle{splncs}
\bibliography{arXiv_finite_automata_with_advice_tapes}

\end{document}